\documentclass[conference]{IEEEtran}
\IEEEoverridecommandlockouts

\usepackage{cite}
\usepackage[utf8]{inputenc}
\usepackage[T1]{fontenc}
\usepackage{titlesec, blindtext, color}
\usepackage{amsthm}
\usepackage{amsmath}
\usepackage{tikz}
\usepackage{amssymb}
\usepackage{amsmath}
\usepackage{mathtools}
\usepackage{color,soul}
\usepackage{dsfont}
\usepackage{cancel}
\usepackage{float}
\usepackage{bbm}
\usepackage[utf8]{inputenc}
\usepackage{enumerate}
\usepackage{framed}
\usepackage{enumitem,color,bbm}

\def\argmin{\mathop{\rm argmin}}
\def\E{\mathbb{E}}

\def\RR{\mathbb{R}}

\newcommand{\scribedby}[1]{}

\newcommand{\mi}{{\textnormal{mid}}}

\newtheorem{lemma}{Lemma}
\newtheorem{theorem}{Theorem}
\newtheorem{remark}{Remark}

\newcommand{\I}{\mathcal{I}}
\newcommand{\R}{\mathcal{R}}
\newcommand{\RI}{\R_{\I}}
\newcommand{\Ball}[2]{\mathbb{B}_{#1}\left(#2\right)}
\newcommand{\bpi}{\boldsymbol{\pi}}
\newcommand{\bp}{\boldsymbol{p}}
\newcommand{\low}{\textnormal{lower}}

\title{
Estimation with Quantized Parameter Side-Information}

\author{\IEEEauthorblockN{Mathis Wetterwald,
\IEEEauthorblockA{LTCI, Telecom Paris IP Paris\\
91120 Palaiseau, France\\
mathis.wetterwald@telecom-paris.fr}}
\and
\IEEEauthorblockN{Ofer Shayevitz,
\IEEEauthorblockA{Tel-Aviv University \\
Tel Aviv, Israel \\
ofersha@eng.tau.ac.il}}
\and
\IEEEauthorblockN{Mich\`ele Wigger,
\IEEEauthorblockA{LTCI, Telecom Paris IP Paris\\
91120 Palaiseau, France\\
michele.wigger@telecom-paris.fr}}}

\date{}

\begin{document}
\maketitle

\begin{abstract}
This paper presents upper and lower bounds on the minimax risk of a parameter estimation problem where the estimator not only observes independent and identically distributed observations but also a quantized version of the parameter, where the quantization has a limitation on the number of levels, but can be optimized otherwise. Our upper and lower bounds have similar decay rates for large number of observation samples and quantization levels.
\end{abstract}

\section{Introduction}
In the classical setting of parameter estimation, one is given $n$ independent and identically distributed (i.i.d.) samples $X^n$ from some unknown distribution $P_\theta$, that belongs to some known parametric family $\{P_\theta\}_{\theta\in \mathcal{J}}$, where in this paper we will consider $\mathcal J$ a bounded interval of $\mathbb R$. The goal is then to provide an estimate $\hat\theta = \hat\theta(X^n)$ for the true underlying parameter $\theta$, while guaranteeing that the guess is optimal in some sense. Typically, the quality of the estimation is measured with respect to  a nonnegative loss function $\ell:\mathcal{J}\times \mathcal{J}\to \RR_+$ (often taken to be quadratic), and in particular by the expected loss, or statistical \textit{risk}, $R_\theta(\hat\theta) = \E_{\theta}[\ell(\theta, \hat\theta)]$. Then, one is usually interested in either 1) uniform bounds on the minimal possible risk of \textit{unbiased estimators} (i.e., satisfying $\E_\theta[\hat\theta] = \theta$), 2) the \textit{Bayesian estimators} that minimize the expected risk $\E_{\theta\sim \pi} [R_\theta(\hat\theta)]$ over some prior $\theta$ on the parameter space, or 3) \textit{minimax estimators} that minimize the worst case risk over the parameter space, i.e., uniformly guaranteeing the minimax risk $R^* = \min_{\hat\theta}\max_{\theta\in \mathcal J} R_\theta(\hat\theta)$. The problem of parameter estimation under these paradigms has been extensively studied in the literature, and is by now very well understood \cite{van2013detection,lecam1986asymptotic,polyanskiy2025information}. Below, we limit our discussion to the minimax setting. 

In this paper, we consider a case where the estimating agent is given access to a quantized version of the parameter $\theta$, prior to seeing the samples. In particular, suppose that the agent can evaluate a $k$-level quantizer $g:\mathcal{J}\to \{1,\ldots, k\}$ of its choosing, and then, after seeing the samples $X^n$ from $P_\theta$, should provide an estimate $\hat\theta(X^n,g(\theta))$. What is the best quantizer $g$ the agent can choose, and what is the effect on the minimax risk? When there are no samples, then this question pertains to simple minimax quantization of the parameter space. For instance, if $\mathcal J=[0,1]$ and under quadratic loss, the optimal quantization is simply uniform and the minimax risk is $1/k^2$. Below, we characterize the minimax risk in the case where $n$ samples from $P_\theta$ are given. It turns out that the minimax risk still improves roughly by a factor of $1/k^2$, but the (essentially) optimal quantizer is somewhat different; it is essentially cyclically uniform with a cycle that corresponds to the standard deviation of the no-quantizer estimation error. 

\textbf{Related work}. Settings of a related flavor have been considered by various authors, especially in the context of distributed parameter estimation with rate constraints, where samples governed by the same parameter are observed by remote parties. Such works include Alswhede and Burnashive \cite{AhlswedeBurnashev90}, Han and Amari \cite{HanAmari95}, Berger and Zhang \cite{zhang2002estimation}, Zhang et al \cite{zhang2013information}, Braverman et al \cite{braverman2016communication}, Hadar and Shayevitz \cite{hadar2019distributed}, and Hadar et al \cite{hadar2019communication}. Loosely speaking, the theoretical difficulties tackled by these works stem either from the fact that each terminal has a small number of samples, or by the fact that the parameter determines the relation between remote samples and cannot be estimated locally. Our work focuses on a somewhat different regime, where one terminal has direct access to the parameter but needs to efficiently encode (quantize) this information to best help the local estimation of the other. 

\section{Problem Setup and Main Result}

Consider a bounded interval $\mathcal{J}\subseteq \mathbb{R}$ and for each $\theta\in \mathcal{J}$ a pmf $P_{\theta}$ over a discrete finite set $\mathcal{X}$. For ease of notation, let $\mathcal{X}=\{1,\ldots, m\}$, for $m$ a given positive integer.  

We consider the two-terminal setup in Figure~\ref{fig:1}. 
\begin{figure}[ht!]
\begin{center} 
\begin{tikzpicture}
    \node at (0,0) [rectangle,draw,ultra thick] (help) {Helper};
    \node at (3,0) [rectangle,draw,ultra thick] (estim) {Estimator};
    \node at (0,1) (theta) {$\theta$};
    \node at (3,1) (Xn) {$X_n$ i.i.d. $\sim P_{\theta}$};
    \node at (3,-1) (hat) {$\hat\theta$};
    \draw[->,thick] (theta) -- (help.north);
    \draw[->,thick] (Xn) -- (estim.north);
    \draw[->,thick] (help.east) -- node[midway,above] {g($\theta$)} (estim.west);
    \draw[->,thick] (estim.south) -- (hat);
    
\end{tikzpicture}
\vspace{0cm}
\caption{Parameter estimation with coded side-information from a perfectly informed helper.}
\label{fig:1}
\end{center}
\end{figure}
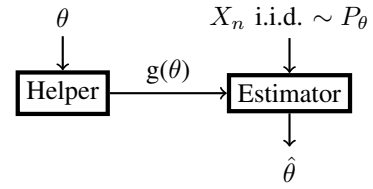
A helper perfectly observes the parameter $\theta\in\mathcal{J}$ that is to be estimated at the estimator. To facilitate the estimator's task, the helper can send a $k$-valued message $M$, for $k\geq 2$,  to the estimator
\begin{equation} 
M= g^{(n,k)}(\theta) \in \{1, \ldots, k\}.
\end{equation}
The estimator observes $n$ independent and identically distributed drawings $X^n=(X_1,\ldots, X_n)$ of the pmf $P_{\theta}$. Based on $X^n$ and $M$ it guesses the parameter $\theta$: 
\begin{equation} 
\hat{\theta}= \phi^{(n,k)} \left(M, X^n\right).
\end{equation} 

We are interested in characterizing the best possible Bayes risk over all helper strategies $g^{(n,k)}$ and estimators $\phi^{(n,k)}$:
\begin{align}
    R_{n,k}^\star :=  \min_{g^{(n,k)}, \phi^{(n,k)}}   \max_{\theta}\E_{P_{\theta}}\left[\left(\theta - \hat{\theta}\right)^2\right].
\end{align}

We shall need the following assumptions.
\begin{itemize} 
\item \textit{Assumption 1:} There exists a real number $D>0$ such that for sufficiently close $\theta, \theta' \in \mathcal{J}$: 
\begin{equation} 
\| P_{\theta} - P_{\theta'}\|_1 \leq D|\theta - \theta'|. 
\end{equation} and there exists an open interval $\mathcal{I}$ such that 
\begin{equation} 
p_{0,\mathcal{I}} := \inf_{\theta \in \mathcal{I}} \min_{x \in \mathcal{X}} P_\theta(x) >0.
\end{equation}
\item \textit{Assumption 2:} 
There exists a real number $F>0$ such that for $\theta, \theta' \in \mathcal{J}$, if $P_\theta$ and $P_{\theta'}$ are sufficiently close: 
\begin{equation} 
\| P_{\theta} - P_{\theta'}\|_1 \geq F |\theta - \theta'|. 
\end{equation} 
\end{itemize} 
Our main result are  upper and lower bounds on    $R_{n,k}^\star$ under above assumptions. Notice that these assumptions are in particular satisfied for the canonical problem of $P_{\theta}$ the Bernoulli-$\theta$ distribution and $\mathcal{J}=[0,1]$.
\begin{theorem}[Converse]\label{thm:1}
Under above Assumption~1, 
 for sufficiently large values of $n$:
\begin{equation} 
C_{\low} \cdot \frac{ 1}{n k^{2}} \leq   R_{n,k}^\star,
\end{equation} 
where 
\begin{equation}
C_{\low} = \frac14(m-1)^{1-m} e^{-\frac{(1+D)^2}{p_{0,\mathcal{I}}}-1}m^{\frac m2}\frac1{(2\pi)^{\frac{m-1}2}}.
\end{equation}
\end{theorem}

\begin{theorem}[Achievability]\label{thm:2}
Under Assumptions~1 and 2, if $k>\frac{D+F}{2F}$ and for $n$ large enough:
\begin{equation}\label{eq:U1}
   R_{n,k}^\star  \leq \frac{8m|\mathcal J|^2}{nk^2}+\frac{m^2\ln(nk^2)}{n\left((k-1/2)F-D/2\right)^2}.
\end{equation}

Moreover, if $\mathcal{J}=[0,1]$ and $P_\theta$ is Bernoulli-$\theta$, then  for $k>2$:
\begin{equation}\label{eq:U2}
   R_{n,k}^\star  \leq \frac{321\ln(k^2)}{n(k/2-1)^{2}} .
\end{equation}
\end{theorem} 


\begin{remark} 
Ignoring the helper information allows the estimator to achieve a minimax risk in the order of $1/n$. Ignoring the data, a simple helper strategy (partitioning the parameter space $\mathcal{J}$ into equally-long subintervals) allows for a minimax risk of $1/k^2$. In that sense, above bounds allow to combine the  benefits of the helper information and the observations $X^n$. 

Notice further that for moderate or large values of $k$ and $n$, the upper bounds in Theorem~\ref{thm:2} strictly improve over above simple strategies.
\end{remark} 

\section{Achievability Proof (Proof of Theorem~\ref{thm:2})}

\subsection{Scheme}

Fix a length 
\begin{equation} \label{eq:L}
L := \sqrt{\frac{\ln(nk^2)}{n}}\frac m{(k-\frac12)\frac F2-\frac D4}. 
\end{equation} 
As $k>\frac{D+F}{2F}$, we have $L>0$. Partition the parameter interval  $\mathcal{J}$ into disjoint subintervals of length $L$, 
where the last subinterval  might be smaller. Denote these intervals by $I_1,...,I_{N}$, for 
\begin{equation}
N:=\left \lceil \frac{|\mathcal{J}|}{L} \right \rceil,
\end{equation} and the mid-points of the subintervals by  $\theta^{\mi}_1,\ldots, \theta^{\mi}_N$. 


We propose the following strategy for the helper and the estimator. 

\textit{Helper:}  If $\theta\in I_j$, the  helper declares  
\begin{equation}\label{eq:help} M =  j\mod k.
\end{equation} 
\textit{Estimator:} Upon  observing  a sequence $X^n$ of  type $\pi_{X^n}$ and  helper message $M$, the estimator declares the mid-point $\theta^{\mi}_j$ that among all the mid-points indicated by the helper message is closest to $\pi_{X^n}$. I.e., it sets 
\begin{equation}\label{eq:cond}
j^*= \argmin_{ j \colon  (j \mod k = M)} \| \pi_{X^n} - P_{\theta^{\mi}_j}\|_1\end{equation}
and declares 
\begin{equation} \label{eq:estim}
\hat{\theta} = \theta^{\mi}_{j^*}.
\end{equation}
If above the argmin is not unique, pick one of the minimizers at random.

In the following subsection, we analyze this scheme and show that it achieves the performances in Theorem~\ref{thm:2}. In fact, for the Bernoulli-$\theta$ result, we will need to modify the choice of $L$ in \eqref{eq:L}.


Figure~\ref{fig:perf} plots the minimax risk $R_{n,k}$ achieved by the proposed helper and estimator strategy for $L=2\sqrt{\ln(nk^2)/(nk^2)}$, $\mathcal J = [0,1]$, $P_\theta = (\theta, \theta-\theta^2,(1-\theta)^2)$ This numerical evaluation confirms that  $R_{n,k}$, when seen as a function of  $k$, is nearly proportional to $\frac{1}{k^2}$ (because the slope of the line in the logspace is close to $-2$). Numerical evaluation thus suggests that $R_{n,k}$ has a decay rate as suggested by our lower bound in Theorem~\ref{thm:1} (a little worse than $\frac{1}{k^2}$, the better the higher $n$ gets). The same numerical evaluation also suggests that  $R_{n,k}$ decays  proportionally to $\frac{1}{n}$, again as suggested by Theorem~\ref{thm:1}.
\begin{figure}[h!]
    \centering
    \includegraphics[width=1\linewidth]{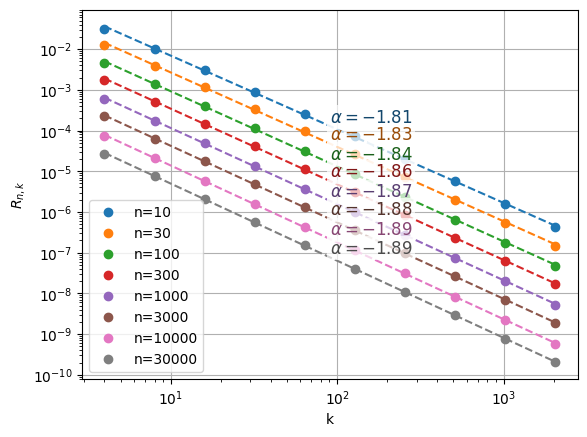}
    \caption{Minimax risk $R_{n,k}$ of the estimator described in~\eqref{eq:help}--\eqref{eq:estim} for varying $n$ and $k$.}
    \label{fig:perf}
\end{figure}

\subsection{Analysis}

To upper bound the expected mean squared error of our proposed strategy $(g^{(n,k)}, \phi^{(n,k)})$, we start by noting that for an  arbitrary parameter $\theta\in\mathcal{J}$, if 
\begin{equation} 
\| \boldsymbol{\pi}_{X^n} - P_{\theta}\|_1 < L \frac{k-\frac12 }{2} F-\frac{LD}{4},
\end{equation} 
then $j^*$ (as defined in Eq.~\eqref{eq:cond}) is the correct interval containing $\theta$ and thus $|\theta_{j^*}^{\mi}-\theta|\leq L/2$. Indeed, for any $j \neq j^*$ such that $j\equiv j^*\mod k$:
\begin{IEEEeqnarray}{rCL}
    \lefteqn{\| \boldsymbol{\pi}_{X^n} - P_{\theta_{j^*}^{\text{mid}}}\|_1 }\nonumber \\
    & \leq & \| \boldsymbol{\pi}_{X^n} - P_{\theta}\|_1+\|  P_{\theta_{j^*}^{\text{mid}}} - P_{\theta}\|_1\\
    & < & L \frac{k-\frac12 }{2} F-\frac{LD}{4}+\frac{LD}2\\
    & = &LF\left(k-\frac12\right)-\left( L \frac{k-\frac12 }{2} F-\frac{LD}{4}\right)\\
    & < & \| P_{\theta} - P_{\theta_j^{\text{mid}}}\|_1-\| \boldsymbol{\pi}_{X^n} - P_{\theta}\|_1 \\
    &\leq & \| \boldsymbol{\pi}_{X^n} - P_{\theta_{j}^{\text{mid}}}\|_1 .
\end{IEEEeqnarray}

We can then write:
\begin{IEEEeqnarray}{rCl}
 \lefteqn{ \mathbb E_{P_\theta}\left[|\hat\theta-\theta|^2\right] } \nonumber\\
    & \leq & P_\theta\left(\|\pi_{X^n}-\theta\|_1> L\frac{k-\frac12}{2} F-\frac{LD}{4}\right)|\mathcal{J}|^2\nonumber\\
    && +P_\theta\left(\|\pi_{X^n}-\theta\|_1\leq L\frac{k-\frac12}{2} F-\frac{LD}{4}\right)\frac{L^2}{4}\\
    & \leq &  8m e^{-\frac{nL^2}{m^2}( (k-1/2) F/2-D/4)^2} |\mathcal{J}|^2+\frac{L^2}{4}\\
    &= &\frac{8m|\mathcal J|^2}{nk^2}+ \frac{m^2\ln(nk^2)}{n((k-1/2)F-D/2)^2},
    \end{IEEEeqnarray}
    where the second inequality holds by  bounding probabilities either by $1$ or as in \cite[Remark to Lemma 2.12]{csiszar_book} and the last equality by the choice of $L$ in \eqref{eq:L}. 
\medskip

\textit{Refined Analysis for the Bernoulli-$\theta$ example leading to \eqref{eq:U2}:}

Modify the length $L$ in the scheme to
\begin{equation} \label{eq:L2}
L=\sqrt{ \frac{2\ln(k)}{n}}(k/2-1) ^{-1},
\end{equation} 
which is positive for $k>2$.
Notice that for this binary example: 
\begin{equation} 
| \boldsymbol{\pi}_{X^n}(0)  - (1- \theta)| = | \boldsymbol{\pi}_{X^n}(1)  -  \theta|,
\end{equation}
and thus $\delta$-typicality as defined in \cite{csiszar_book} is equivalent to the condition $| \boldsymbol{\pi}_{X^n}(1)  -  \theta|\leq \delta$.

To upper bound the expected mean squared error of our proposed strategy $(g^{(n,k)}, \phi^{(n,k)})$, we start by noting that for an  arbitrary parameter $\theta\in\mathcal{J}$, if 
\begin{equation} 
| \boldsymbol{\pi}_{X^n}(1) - \theta| \leq L \frac{k-2 }{2},
\end{equation} 
then $j^*$ is the correct interval containing $\theta$ and thus $|\theta_{j^*}^{\mi}-\theta|\leq L/2$. In a similar way, for $i=1,2,\ldots$, if 
\begin{IEEEeqnarray}{rCl}
| \boldsymbol{\pi}_{X^n}(1) - \theta| & \in& \left( L\frac{k-2}{2}+ (i-1)L k, \; \right.\nonumber \\
&& \qquad  \left. L\frac{k-2}{2}+ iL k\right],
\end{IEEEeqnarray}
then $j^*$ represents the interval with $g(j^*) = g(\theta)$ that is at most $i$-th closest  to the interval ${I}_{g(\theta)}$,   and with our strategy the squared error is bounded by
\begin{equation} 
|\hat{\theta}-\theta|^2 \leq \left( \frac{L}{2} + i L k\right)^2.
\end{equation}
\allowdisplaybreaks[4]
We can then write, with $|\pi_{X^n}(1)-\theta| =d_\theta$:
\begin{IEEEeqnarray}{rCl}
   \lefteqn{ \mathbb E_{P_\theta}\left[|\hat\theta-\theta|^2\right]} \ \nonumber \\ 
    & \leq &P_\theta\left(d_\theta\leq L \frac{k-2 }{2}\right)\frac{L^2}{4} \nonumber \\
    & &+2 \sum_{i=1}^{N-1} P_{\theta}\Big(d_\theta \in\Big(L\frac{k-2}{2} +(i-1) Lk ,\nonumber \\
    &&  \hspace{2.3cm} L\frac{k-2}{2} + i L k \Big]\Big) \left( \frac{L}{2} + i L k\right)^2  \IEEEeqnarraynumspace\\
    & \stackrel{(a)}{\leq}& \frac{L^2}{4}   + 2 P_\theta\Big(d_\theta > L\frac{k-2}{2}\Big)\left( \frac{L}{2} +  L k\right)^2\nonumber \\
    && \hspace{0cm} +2\sum_{i=2}^{N-1}P_\theta\Big(d_\theta > L\frac{k-2}{2}+(i-1) Lk\Big) \nonumber \\
    && \hspace{0.8cm}\times\underbrace{ \left( \left(\frac{L}{2}+  iL{k}\right)^2 - \left(\frac{L}{2}+  (i-1)L{k}\right)^2  \right)}_{ =L^2k(1+(2i-1)k)}\IEEEeqnarraynumspace \\
    & \stackrel{(b)}{\leq}&  \frac{L^2}{4}+32me^{-nL^2(k/2-1)^2}(Lk)^2  \nonumber\\
    &&+16mL^2 k \sum_{i=2}^{N-2} (1-k+2ik)e^{-n L^2  (\frac k2-1+(i-1)k)^2}   \\
      &  \stackrel{(c)}{\leq}&  65L^2+ 16mL^2 k\sum_{i=2}^{N-2} (2ik)e^{-n L^2  (i-1)^2k^2}  \nonumber  \\
      &&+16mL^2k\times 2k e^{-nL^2(k/2-1)^2}\\
          &  \stackrel{(d)}{\leq}&   65L^2\hspace{-1mm}+\hspace{-1mm}32mL^2\left(1+ k^2\sum_{i=2}^{\infty} i e^{-n L^2  (i-1)^2k^2}\right)
      \label{eq:d}
\end{IEEEeqnarray}
where we notice that the infinite sum in \eqref{eq:d}  converges because $e^{-n L^2 k^2 }<1$.
In above sequence of inequalities:
\begin{itemize} 
\item $(a)$ holds by writing  $P_{\theta}(|\pi_{X^n}(1)- \theta| \in(a,b])=P_{\theta}(|\pi_{X^n}(1)- \theta| >a)- P_{\theta}(|\pi_{X^n}(1)- \theta|>b)$, rearranging the terms in the sum and dropping negative terms; 
\item $(b)$ holds by bounding the probabilities using  \cite[Remark to Lemma 2.12]{csiszar_book};
\item $(c)$  holds by using the value of $L$ to simplify the first exponential term,  by  getting $i=1$ out of the sum,  and by dropping for the other summands the negative term $(1-k)$ as well as the positive term $k/2-1$ in the exponent;
\item $(d)$ holds by adding nonnegative terms to the sum, and computing the case $i=1$.
\end{itemize}

Plugging  the choice of $L$ in \eqref{eq:L2},  into  (\ref{eq:d}) yields:
\begin{IEEEeqnarray}{l}
 \lefteqn{\mathbb E_{P_\theta}\left[|\hat\theta-\theta|^2\right]}\nonumber\\ \leq  \frac{\ln(k^2)}{n(k/2-1)^{2}} \left(65+  32m(k^2A+1)\right),
\end{IEEEeqnarray}
where here $m=2$ and for $k\geq 2$, 
\begin{IEEEeqnarray}{rCl} 
A&=& \sum_{i=2}^{\infty} i e^{-n L^2  (i-1)^2k^2}<\sum_{\ell=0}^\infty (\ell+2) \left(\frac{1}{k^2}\right)^{(\ell+1)^2} \\&\leq&\frac1{k^2}\sum_{\ell=0}^\infty (\ell+2)\left(\frac{1}{4}\right)^{ \ell^2} \nonumber
< \frac{2.77}{k^2},
\end{IEEEeqnarray}
and can be bounded  by $3/k^2$.

\section{Converse Proof (Proof of Theorem~\ref{thm:1})}

Before presenting the proof, we make some definitions and state an auxiliary lemma. 
Let $\Ball{\boldsymbol{p}}{r}$ be the $\mathbb{L}_1$-ball of type-vectors $\boldsymbol{\pi}\in \mathcal{P}_n$ of radius $r$ around a probability vector $\boldsymbol{p} \in \mathcal{P}$: 
\begin{equation}
\Ball{\bp}{r} \triangleq \{\bpi\in\mathcal{P}_n \colon\;  \| \bp- \bpi \|_1 \ \leq r\}
\end{equation}
Let  $\mathcal{I}$ be any subinterval of the parameter space  $ \mathcal{J}$ of length $|\mathcal{I}|=n^{-1/2}$  with 
\begin{equation}\label{eq:p0}
 p_0\triangleq    
\min_{\theta\in\mathcal{I}}\min_{x\in \mathcal{X}} P_{\theta}(x) > 0.
\end{equation}
Define the region of types
\begin{equation}
\RI = \bigcup_{\theta' \in \mathcal{I}}\ \Ball{P_{\theta'}}{ \frac1{\sqrt n}} . 
\end{equation}


We then have the following lemma:
\begin{lemma}\label{lemma1}
\begin{enumerate}
\item For sufficiently large $n$, the number of $n$-types in  region $\RI$ is \begin{equation}
    |\RI | \geq  (m-1)^{-(m-1)}\cdot \sqrt n^{m-1}.
\end{equation}
\item For all $\theta'\in \I$ and $\bpi \in \RI$:
    \begin{equation}
         P_{\theta'}(\bpi)\geq \frac{C}{\sqrt{n}^{m-1}}=:P_{\min} ,
    \end{equation}
    where $C=e^{-(1+D)^2 /p_0-1}m^{\frac m2}\frac1{(2\pi)^{\frac{m-1}2}}$

\end{enumerate}
\end{lemma}

\begin{proof}See Appendix~\ref{app:lem1}. \end{proof}

We can now prove the desired lower bound in the theorem. Pick $\theta_1, \theta_2\in \mathcal{I}$ with same helper message, $g(\theta_1)=g(\theta_2)=j$ and distance:
\begin{equation} 
|\theta_1-\theta_2| \geq \frac{1}{k\sqrt{n}}.
\end{equation}
Notice that such parameter values must exist because the helper function takes  only $k$ possible values and because $\I$ is of length $1/\sqrt{n}$.

The desired lower bound is then obtained through a modified 2-point Le-Cam method where the analysis is restricted to the set $\RI$:
\begin{align}
   2 R_{k,n}^\star & \geq   \mathbb E\left[|\theta_1-\hat\theta_1|^2\right] +\mathbb E\left[|\theta_2-\hat\theta_2|^2\right] \\ 
   &\geq    \sum_{\bpi\in \RI} P_{\min}\cdot \Big(|\theta_1-\hat\theta(j,\bpi)|^2 \nonumber \\
   & \hspace{2cm} + |\theta_2-\hat\theta(j,\bpi)|^2\Big)\\
    & \geq \frac12 |\RI| P_{min}\cdot|\theta_1-\theta_2|^2\\
    & \geq \frac12(m-1)^{-m+1} C|\theta_1-\theta_2|^2 \\
    & \geq  \frac12(m-1)^{-m+1} C \cdot  \frac{1}{k^{2} n} .
\end{align}

    \section*{Acknowledgements}
The work of M. Wetterwald and M. Wigger has been supported by the ERC under grant agreement $101125691$.
\begin{appendices}
\section{Proof of Lemma~\ref{lemma1}}\label{app:lem1}

\textit{Proof of 1)}: 
By definition of  $\RI$:  
\begin{equation}
    \Ball{P_{{\theta}'}}{\frac1{\sqrt n}} \subseteq \RI
    \label{eq:balls}
\end{equation}
Moreover, by Assumption \eqref{eq:p0}, 
for sufficiently large $n$, the number of $n$-types in   ball $\Ball{P_{\theta'}}{r}$ is  lower bounded by $(r/(m-1))^{m-1}$. Setting $r=n^{-1/2}$, this concludes the proof of 1). 


\textit{Proof of 2):} For any $\theta \in \I$ and type $\bpi \in \mathcal{R}_I$: 
\begin{equation}  P_{\theta}(\bpi) = \begin{pmatrix}
        n \\ n \bpi(1),...,n\bpi(m)
    \end{pmatrix}\prod_{i\in \mathcal{X}}P_{\theta}(i) ^{n\bpi(i)}.\end{equation}
    We use Stirling's  and Herbert Robbins' bounds to write
    \begin{equation}
        e^{\frac1{12n+1}}\frac{n^n}{e^n}\sqrt{2\pi n}<n!<e^{\frac1{12n}}\frac{n^n}{e^n}\sqrt{2\pi n}
    \end{equation}
 and thus to conclude:
    \begin{align}
        P_{\theta}(\bpi) &\geq \frac{e^{\frac1{12 n+1}}}{\prod_{i=1}^me^{\frac1{12 n\bpi(i)}}}\cdot \frac{ n^ n}{\prod_{i=1}^m ( n\bpi(i))^{ n\bpi(i)}} \nonumber \\
        &\qquad \cdot  \frac{\sqrt{2\pi  n}}{\prod_{i=1}^m\sqrt{2\pi  n\bpi(i)}}\prod_{i=1}^mP_{\theta}(i) ^{n\bpi(i)} \\
        &\geq \frac{e^{\frac1{12 n+1}}}{ e^{\frac{m}{12 (np_0-\sqrt{n})}}}\cdot \frac{1}{\prod_{i=1}^m \left(\frac{ \bpi(i)}{ P_{\theta}(i)}\right)^{ n \bpi(i)}} \nonumber \\
        & \qquad \cdot \frac{1}{\sqrt{2\pi}^{m-1}}\cdot \sqrt{\frac{ n}{\prod_{i=1}^m n \bpi(i)}}
    \end{align}
where we used that  $\sum_i \bpi(i) = 1$ and 
\begin{equation} 
\frac{e^{\frac1{12 n+1}}}{\prod_{i=1}^me^{\frac1{12 n\bpi_i}}} \geq \frac{e^{\frac1{12 n+1}}}{ e^{\frac{m}{12 (np_0-\sqrt{n})}}}
.
\end{equation} 
Moreover,  since the geometric mean is upper bounded by the arithmetic mean, 
\begin{equation} \prod_{i=1}^m \bpi(i)\leq \left(\frac{\sum_{i=1}^m \bpi(i)}{m}\right)^m = m^{-m} \end{equation} and thus :
    \begin{equation}
        \sqrt{\frac{ n}{\prod_{i=1}^m n \bpi(i)}} \geq  n^{-\frac{m-1}2} m^{\frac m2}.
    \end{equation}
    By the triangle inequality and Assumption 1, 
    we  have for any $\theta\in \mathcal{I}$ and $\bpi \in \RI$:
\begin{IEEEeqnarray}{rCl}
\| {\bpi} - P_{{\theta}} \|_1 & \leq &\| {\bpi} - P_{\theta'} \|_1 +\|  P_{\theta'} -P_{{\theta}} \|_1\\
& \leq &
 \frac{1}{\sqrt{n}} + D| \theta' -{\theta}| \\
 & \leq&  \frac{1+ D }{\sqrt{n}},\label{eq:bound}
 \end{IEEEeqnarray}
 where $\theta'$ is the parameter so that $\bpi \in \Ball{P_{\theta'}}{1/\sqrt{n}}$. 

 We therefore have:   \begin{IEEEeqnarray}{l}
       \lefteqn{ \log \left(\prod_{i=1}^m \left(\frac{ \bpi(i)}{ P_{\theta}(i)}\right)^{ n \bpi(i)}\right) } \nonumber \\=\sum_{i=1}^m  n \bpi(i) \log\left(\frac{ \bpi(i)}{ P_{\theta}(i)}\right)\\
         = \sum_{i=1}^m  n \bpi(i) \log\left(1+\frac{ \bpi(i)- P_{\theta}(i)}{ P_{\theta}(i)}\right)\\
        \leq \sum_{i=1}^m  n(( \bpi(i)- P_{\theta}(i))+ P_{\theta}(i))\frac{ \bpi(i)- P_{\theta}(i)}{ P_{\theta}(i)}\\
         =  n\sum_{i=1}^m( \bpi(i)- P_{\theta}(i)) +  n\sum_{i=1}^m\frac{( \bpi(i) - P_{\theta}(i))^2}{ P_{\theta}(i)}\\
        \leq 0+ (1+ D)^2  /p_0.
    \end{IEEEeqnarray}
    
    Combining above observations, and noting that for $n$ such that $n\geq \frac{4}{p_0^2}$ and $n\geq \frac{m^2}{144}$, 
    \begin{equation} 
    \frac{e^{\frac1{12 n+1}}}{ e^{\frac{m}{12 (np_0-\sqrt{n})}}}
 \geq\frac1{e^{\frac{m}{12 \sqrt{n}}}}\geq e^{-1},
    \end{equation} 
    we obtain that for sufficiently large values of $n$:
    \begin{equation}
         P_{\theta'}(\bpi)\geq e^{E-1}m^{\frac m2}\frac1{(2\pi)^{\frac{m-1}2}}n^{-\frac{m-1}2}
    \end{equation}
    for $E=-(1+D)^2 /p_0.$ This proves part 2). 

\end{appendices}
\bibliographystyle{ieeetr}
\bibliography{our_bib}
    
 \end{document}